\documentclass[12pt]{article} 
\usepackage[sectionbib]{natbib}
\usepackage{array,epsfig,fancyheadings,rotating}
\usepackage[]{hyperref}  
\usepackage{sectsty, secdot}
\sectionfont{\fontsize{12}{14pt plus.8pt minus .6pt}\selectfont}
\renewcommand{\theequation}{\thesection\arabic{equation}}
\subsectionfont{\fontsize{12}{14pt plus.8pt minus .6pt}\selectfont}

\usepackage{amsmath}
\usepackage{amssymb}
\usepackage{amsfonts}
\usepackage{multirow}
\usepackage{amsthm}
\usepackage{algorithm2e}
\makeatletter
\renewcommand{\@algocf@capt@plain}{above}
\makeatother

\usepackage{enumerate}
\usepackage{array,color}
\usepackage{xcolor}
\usepackage{extarrows}
\usepackage{cases}
\usepackage{mathrsfs}
\usepackage{graphicx}
\usepackage{float}
\usepackage{subfigure}
\usepackage{diagbox}
\usepackage{mathtools}
\usepackage{geometry}
\usepackage{amsmath}

\newtheorem{theorem}{Theorem}

\theoremstyle{definition}

\DeclarePairedDelimiter\ceil{\lceil}{\rceil}

\newcommand{\one}{{\mathbf 1}}
\newcommand{\I}{{\mathbf I}}

\newcommand{\calh}{{\mathcal H}}

\def\RR{{\mathbf{ R}}}

\newcommand{\beq}{\begin{equation}}
\newcommand{\eeq}{\end{equation}}
\newcommand{\beqn}{\begin{eqnarray}}
\newcommand{\eeqn}{\end{eqnarray}}

\newcommand{\rv}{\mbox{Var}}

\newcommand{\E}{\mbox{E}}

\begin{document}


\renewcommand{\baselinestretch}{2}


\markboth{\hfill{\footnotesize\rm Radu V. Craiu  AND Xiao-Li Meng} \hfill}
{\hfill {\footnotesize\rm L-lag coupling with Control Variates} \hfill}

\renewcommand{\thefootnote}{}
$\ $\par


\fontsize{12}{14pt plus.8pt minus .6pt}\selectfont \vspace{0.8pc}
\centerline{\large\bf Double Happiness: Enhancing the Coupled  }
\vspace{2pt} 
\centerline{\large\bf Gains of L-lag Coupling  via Control Variates}
\vspace{.4cm} 
\centerline{Radu V. Craiu and Xiao-Li Meng} 
\vspace{.4cm} 
\centerline{\it University of Toronto and Harvard University}
 \vspace{.55cm} \fontsize{9}{11.5pt plus.8pt minus.6pt}\selectfont


\begin{quotation}
\noindent {\it Abstract:}
{\bf The recently proposed L-lag coupling for unbiased Markov chain Monte Carlo (MCMC)  calls for a joint celebration by MCMC practitioners and theoreticians. For practitioners, it circumvents the thorny issue of deciding the burn-in period or when to terminate an MCMC sampling process, and opens the door for safe parallel implementation. For theoreticians, it provides a powerful tool to establish elegant and easily estimable bounds on the exact error of an MCMC approximation at any finite number of iterates. A serendipitous observation about the bias-correcting term  leads us to introduce naturally available control variates into the L-lag coupling estimators. In turn, this extension enhances the coupled gains of L-lag coupling,  because it results in more efficient unbiased estimators, as well as a better bound on the total variation error of MCMC iterations, albeit the gains diminish as L increases. Specifically, the new upper bound  is theoretically guaranteed to never exceed the one given previously.  We also argue that L-lag coupling represents a coupling for the future, breaking from the coupling-from-the-past  type of perfect sampling, by reducing the generally unachievable requirement of being \textit{perfect} to one of being \textit{unbiased}, a worthwhile trade-off for ease of implementation in most practical situations.   The theoretical analysis is supported by numerical experiments that show tighter bounds and a gain in efficiency when control variates are introduced.}\\

\vspace{9pt}
\par
\end{quotation}\par

\def\thefigure{\arabic{figure}}
\def\thetable{\arabic{table}}

\renewcommand{\theequation}{\thesection.\arabic{equation}}

\fontsize{12}{14pt plus.8pt minus .6pt}\selectfont

\section{If Being Perfect Is Impossible, Let's Try Being Unbiased}\label{sec:intro}
\subsection{Perfect Coupling  -- Too Much To Hope For?}
We thank Pierre Jacob and his team for a series of  articles \citep[e.g.,][]{jacob2020unbiased, jacob2019smoothing, heng2019unbiased, biswas2019estimating} that revitalized our experience \citep[e.g.,][]{murdoch2001,meng:multistage-backwards, craiu2011perfection, stein2013practical} of working on  coupling from the past \citep[CFTP;][] {propp-wilson:exact-sampling,propp-wilson} and, more generally, perfect sampling.   The clever ``cross-time coupling" idea of \cite{glynn2014exact}, which can be considered a form of coupling for the future (CFTF),   allows us to move away from the CFTP framework, which became popular  around the turn of the century with its promise of providing perfect/exact Markov chain Monte Carlo (MCMC) samplers \citep[e.g., see the annotated bibliography of][]{exact-bibliography}. However,  research progress on perfect or exact samplers has slowed  significantly since then, because they are very challenging, if not impossible, to develop for many routine Bayesian computational problems \citep[e.g., see][]{murdoch2001}. 

In its most basic form, a CFTP-type perfect sampler couples a Markov chain $\{X_t, t\ge 0\}$ with itself, but from different starting points, and runs two or more chains until they coalesce at a time $\tau$. This apparent convergence does not guarantee, in general, that $X_{\tau}$ is from the desired stationary distribution $\pi(x)$. By shifting the entire chain to ``negative time"(i.e., the past), $\{X_t,  t\le 0\}$,  \cite{propp-wilson:exact-sampling} have shown that if we follow this coalescent chain until it reaches the present time, that is, $t=0$, then the resulting $X_0$ will be exactly from $\pi(x)$. Perhaps the most intuitive way to understand this scheme is to realize that running a chain from its infinite past ($t=-\infty$) to the present ($t=0$) is mathematically equivalent to running the chain from the present ($t=0$) to the infinite future ($t=+\infty$). The CFTP is a clever way of realizing this seemingly impossible task, relying on the fact that if the coalescence occurs regardless of how we have chosen the starting point, then the chain has ``forgotten" its origin, and hence settled in the perfect asymptotic distribution.

However, being perfect is never easy, especially in the mathematical sense. No error of any kind is allowed, and this requirement has manifested in two ways that greatly limit the practicality of perfect sampling. First, constructing a perfect sampler, especially for distributions with continuous and unbounded state spaces---which are ubiquitous in routine statistical  applications---is a very challenging task in general, despite its great success for problems with some special structures, such as  certain monotonic properties \citep[see][]{MR1978833,MR1906013, MR2052900, sw-hub,ensor2000simulating, MR2052900,mur-que}.
Second, even if a perfect sampler is devised, it can be excruciatingly slow, because it refuses to deliver an output until it can guarantee its perfection,  and one must devise problem-specific strategies to speed this up 
\citep[e.g.,][]{MR1699662, MR2016872, moller:conditional,MR2016872,MR2154791}.

\subsection{Unbiased Coupling -- A New Hope?}
A  relaxation of the exact sampling paradigm with important practical consequences  has been  proposed by \cite{glynn2014exact} and \cite{glynn2016exact}, who put forth strategies for the exact estimation of integrals using  MCMC.  The difference between exact sampling and exact estimation is a large conceptual leap that allows us to bypass most of the difficulties of perfect sampling, while maintaining some of its important benefits. Building on the work of Glynn and co-authors,  the L-lag coupling of  \cite{biswas2019estimating} and \cite{jacob2020unbiased}  aims to deliver unbiased estimators of $\E[h(X_\pi)]$, for any (integrable) $h$, where $X_\pi$ denotes a random variable defined by $\pi(X)$. 
One may question if this is  really a weaker requirement because the fact that
$\E[h(X_\pi)]=\E[h(X_{\tilde\pi})]$ for all (integrable) $h$ immediately implies that $\pi(X)=\tilde \pi(X)$ (almost surely). This is where the innovation of L-lag coupling lies, because it does not couple a chain with itself from two or more starting points [e.g., two extreme states, as with monotone coupling; see \cite{propp-wilson:exact-sampling}]. Instead, it couples two  chains that have the same transition probability and  start from the same starting point or, more generally, the same initial distribution $\pi_0$, but are time-shifted by an  integer lag, $L >0$.  

To illustrate, consider the case of $L=1$,  which was the focus of \cite{jacob2020unbiased}.  Two chains ${\cal X}=\{X_t, t\ge 0\}$ and ${\cal Y} = \{Y_t, t\ge 0\}$ are coupled in such a way that both of them have the same transition kernel (and, hence, the same target stationary distribution), and there exists with probability one a finite stopping time $\tau$, such that $X_{t}=Y_{t-1}$, for all $t \ge \tau$.  This construction allows them to show that the following estimator based on both $\cal X$ and $\cal Y$,
\beqn
H_k({\cal X,Y}) = h(X_k)+\sum_{j=k+1}^{\tau -1} [h(X_j)-h(Y_{j-1})], 
\label{hk}
\eeqn
is an unbiased estimator for $\E[h(X_\pi)]$, for any $k\ge 0$ (under mild conditions). Heuristically,  this is because the sum in (\ref{hk}) is the same as $\sum_{j=k+1}^\infty[h(X_j)-h(Y_{j-1})]$, because any term with $j\ge \tau$ must be zero, by the coupling scheme.  Furthermore, for the purpose of calculating expectations, we can replace $h(Y_{j-1})$
with $h(X_{j-1})$, for any $j$, because $X_{j-1}$ and $Y_{j-1}$ have identical distributions, by construction.  However, $h(X_k)+\sum_{j=k+1}[h(X_j)-h(X_{j-1})]$ is nothing but $\lim_{t\rightarrow \infty} h(X_t)$, which has the same distribution as $h(X_\pi)$. 

The cleverness of constructing an estimator based on \textit{both} ${\cal X}$ and ${\cal Y}$ to ensure $\E[H_k({\cal X, Y})]=\E[h(X_\pi)]$, for any $h$, bypasses the requirement that $X_\tau$ itself must be perfect.
The series of illustrative and practical examples in \cite{jacob2020unbiased}  and in \cite{jacob2019smoothing}, \cite{heng2019unbiased}, and \cite{biswas2019estimating} provide good evidence of the practicality of this approach.  The use of parallel computation for estimating $I= \E_{\pi}[h(X)]$  supports using 
$\E[\E_{\pi} [h(x)| {\cal{U}}_j]]$, where the inner expectation is the estimate obtained from the $j$th parallel process, ${\cal{U}}_j$, and the outer mean averages over all processes. However, if each inner mean is a biased estimator for $I$, then the accumulation of errors can be seriously misleading. This has been documented in the Monte Carlo literature extensively,  for instance, in \cite{glynn1991analysis} and \cite{nelson2016some}.
Hence, unbiased MCMC designs allow  one to take full advantage of parallel computation strategies, without having to worry about the accumulation of bias as the number of parallel processes increases.

\subsection{Using Control Variates  -- Even Higher Hope?}
The expression (\ref{hk}) also opens a path to explore further improvements, and that is the starting point of our exploration. In \cite{cmdisc},  we noticed that (\ref{hk}) can be expressed equivalently as
\beqn
H_k({\cal X,Y}) = h(X_{(\tau-1)\vee k}) + \sum_{j=k}^{\tau-2}[h(X_j)-h(Y_j)],
\label{hkb}
\eeqn
where $A\vee B=\max\{A, B\}$. Expression (\ref{hk}) renders the insight underlying \cite{jacob2020unbiased}, which is that $H_k({\cal X,Y})$ achieves the desired unbiasedness by providing a \textit{time-forward bias correction} to $h(X_k)$, whenever $\tau> k+1$;  hence coupling for the future. (No correction is needed when $\tau\le k+1$.) The dual expression (\ref{hkb}) indicates that $H_k({\cal X,Y})$ can also be viewed as a \textit{time-backward bias correction} to $h(X_{\tau-1})$ for its imperfection, because $k <\tau -1$. 

Most intriguingly, each correcting term $\Delta_j\equiv h(X_j)-h(Y_j)$ in (\ref{hkb}) has mean zero, by the construction of $\{\cal X, Y\}$. However, the sum $\sum_{j=k}^{\tau-2}[h(X_j)-h(Y_j)]$ does not necessarily have mean zero, because $\tau$ is random and it depends critically on $\{\cal X, Y\}$. Indeed, if this sum had mean zero, then $X_{(\tau-1)\vee k}$ would have been a perfect draw from $\pi(X)$, because then $\E[h(X_{(\tau-1)\vee k})]=\E[h(X_\pi)]$, for any (integrable) $h$, which would imply that $X_{(\tau-1)\vee k}\sim \pi$.  

However, the fact that $\E(\Delta_j)=0$ suggests that we can use any linear combination of $\Delta_j$ as a \textit{control variate} for $H_k({\cal X,Y})$. Using control variates to reduce estimation errors is a well-known technique in the literature on improving MCMC samplers and estimators by using efficiency swindles, such as antithetic and control variates, 
Rao--Blackwellization, and so on, some of which we have explored in the past \citep[e.g.,][]{van2001art, craiu2001antithetic, craiu2005multiprocess, craiu2007acceleration, yu2011center}. For example, for any finite constant $\eta>k+1$, the estimator 
\beqn
H_k^*({\cal X,Y};\eta) 
= H_k({\cal X,Y})   -\sum_{j=k}^{\eta-2}\Delta_j
= h(X_{(\tau-1)\vee k}) + \sum_{j=k}^{\tau-2}\Delta_j -\sum_{j=k}^{\eta-2}\Delta_j
\label{hkc}
\eeqn
shares the mean of $H_k({\cal X,Y})$, but can have a smaller variance, with a judicious choice of $\eta$. Intuitively, this reduction of variance is possible because of the potential partial cancellation (on average) of the  $\Delta_j$ terms in the last two summations in (\ref{hkc}).

Indeed, Section~\ref{sec:CV}  investigates a more general class of control variates, and derives the optimal choice by establishing the minimal upper bound within the class on the total variation distance between the target $\pi$ and $\pi_k$, the distribution of
$X_k$. This leads to both an improved theoretical bound over that of \cite{biswas2019estimating}, as reported in Section~\ref{sec:CV}, as well as a more efficient estimator than (\ref{hk}) owing to a 	 parallel implementation. Section~\ref{sec:estim} describes the estimation methods and algorithms, and Section~\ref{sec:prac} provides examples and illustrations of both kinds of gains.  Section~\ref{sec:fut} discusses some future work.    

\section{Theoretical Gains from Incorporating Control Variates  }
\label{sec:CV}

\setcounter{equation}{0}

\subsection{L-lag Coupling:  An Elegant and Powerful Method}

The scheme of L-lag coupling extends the coupling of $\{X_k, Y_{k-1}\}$ to the more general form of the coupling of $\{X_k, Y_{k-L}\}$,  for some fixed $L\ge 1$, as detailed in \cite{biswas2019estimating}. The significance of this extension can be best understood by expressing the L-lag coupling idea in its mathematically equivalent form of seeking $\tau_L$ such that $X_{k+L}= Y_{k}$, for all $k\ge \tau_L$, and letting $L\rightarrow \infty$ while keeping $k$ fixed.  Heuristically, it is then clear that the larger $L$, the closer the distribution of $Y_{\tau_L}$  is to the target, because $X_{\tau_L+L}$ should converge to $X_\infty\sim \pi$ as $L\rightarrow\infty$, and  ${\cal X}$ and ${\cal  Y}$ share the same target $\pi$. 

Indeed, by extending (\ref{hk}) to a general $L$, \cite{biswas2019estimating} show that (under mild regularity conditions) the total variation distance between $\pi_k$, the distribution of $X_k$, and $\pi$ is  bounded by a very simple function of $\tau_{L}$ and $(k, L)$:
\beqn\label{eq:biswas}
d_{\rm TV}(\pi_k, \pi) \le \E[J_{k,L}],   \qquad {\rm with}\quad  J_{k,L}=\max\left\{0,  \ceil{\frac{\tau_L-L-k}{L}}\right\},
\eeqn
where $\ceil{a}$ denotes the smallest integer that is no less than $a$. We can clearly see the impact of increasing $L$ or $k$, because larger values of either of them make it more likely that $\tau_L-L-k <0$, and hence $J_{k,L}=0$. Perhaps a  clear demonstration of this fact is when $\tau_L$ follows a geometric distribution
with success probability $p$ and state space $\{L+i, \  i\ge 0\}$ (because $\tau_L\ge L$, by definition)
or, equivalently, $\delta=\tau - (L-1)\sim {\rm Geo}(p)$. Then, letting $q=1-p$, we have \citep[see][]{biswas2019estimating}
\beqn\label{eq:j-mean}
d_{\rm TV}(\pi_k, \pi) \le \E[J_{k,L}] =
\frac{q^{k+1}}{1- q^L}.
\eeqn
We see that the bound is a decreasing function of both $k$ and $L$, though it decreases much faster with $k$, which controls the rate of convergence, than it does with $L$, which controls only the (constant) scaling factor. We also observe that the bound can be trivial, because it can be larger than one for small $k$ and/or $L$, whereas 
$d_{\rm TV}$ cannot, suggesting there is room for improvement.  Nevertheless, \eqref{eq:biswas} is a remarkable bound because it encodes all the intricacies relevant for the convergence speed of ${\cal X}$, including the choice of $X_0$, into a univariate (truncated) coupling time
$J_{k,L}$. In the case of (\ref{eq:j-mean}), the bound also immediately establishes the geometric ergodicity of ${\cal X}$, and provides a rather practical way to assess the bound by estimating $p$ or, more generally, by assessing $J_{k,L}$ directly, say, from a parallel implementation (see Section~\ref{sec:estim}). 

It is perhaps even more remarkable to see that the left-hand side of (\ref{eq:biswas}) is a property of the marginal chain ${\cal X}$ (and, equivalently, of the ${\cal Y}$ chain), but its right-hand side depends on the construction of the joint chain $\{\cal X, Y\}$.  This suggests that we can seek improvement by better coupling. Furthermore, as we establish below, even without changing the coupling scheme, we can still obtain better bounds by using more efficient estimators than (\ref{hk}).     

For a general $L$,  the forward-correction expression in (\ref{hk}) becomes \citep{biswas2019estimating} 
\beqn
H_{k,L}({\cal X,Y}) = h(X_k)+\sum_{j=1}^{J_{k,L}} \left[h(X_{k+jL})-h(Y_{k+(j-1)L})\right], 
\label{eq:hkLf}
\eeqn
and it is easy to verify that the backward-correction expression (\ref{hkb}) takes the form  
\beqn\label{eq:hkLb}
H_{k,L}({\cal X,Y}) 
= h(X_{k+LJ_{k,L}}) + \sum_{j=0}^{J_{k,L}-1} \left[h(X_{k+jL})-h(Y_{k+jL})\right].
\eeqn

\noindent{\it Remark 1:} 
The (random) subscript in $X_{k+JL}$ cannot be reduced to $(\tau-L)\vee k$ when $L>1$,  the most obvious extension of the index $(\tau-1)\vee k$ in  (\ref{hkb}).  This is because $k+J_{k,L}L\ge (\tau-L)\vee k$, but the inequality can be strict when $\tau> k+L$. For example,  if $\tau=L+k+M$, where $M$ is a  positive integer less than $L$ (which does not exist when $L=1$), $k+J_{k,L} L=k+L$, but  $(\tau-L)\vee k=k+M$.  

\noindent{\it Remark 2:} Whereas (\ref{eq:hkLf}) and (\ref{eq:hkLb}) are equivalent as equalities, they may lead to different inequalities depending on how we bound their respective right-hand sides. This is both a bonus and a trap, as we discuss below.   

\subsection{Deriving the Optimal Bound over Choices of Control Variates}\label{sec:opt}

For notational simplicity, we drop the variables $k,L$ from the notation of $J_{k,L}$, and we
let $\Delta_{k, j}=h(X_{k+jL})-h(Y_{k+jL})$. Then, we know $\Delta_{k,j}$ has mean zero for any $\{k, j\}$ and $L$.  This means that for any random sequence $\vec{\eta}\equiv\{\eta_j, j\ge 1\}$ such that: (A) it is independent of  $\{{\cal X, Y}\}$, 
and (B) $\sum_{j=1}\E_{\vec{\eta}}|\eta_j|<\infty$,  we can use $C_\eta=\sum_{j\ge 1} \eta_j\Delta_{k, j}$ as a control variate for $H_{k, L}\equiv H_{k, L}({\cal X, Y})$, because $\E[C_\eta]=0$. That is, 
\beqn\label{cont}
\tilde H_{k,L}^{(\vec\eta)}({\cal X, Y})= H_{k, L}({\cal X, Y}) - \sum_{j\ge 1}\eta_j \Delta_{k, j} 
\eeqn
is also an unbiased estimator of $\E[h(X_\pi)]$ with a smaller variance than \eqref{eq:hkLb}. Next, we examine how to choose  $\eta$.

To choose $\vec\eta$,  instead of  minimizing $\rv\left[\tilde H_{k,L}^{(\vec\eta)}\right]$, which is not an easy task and will also likely produce an $h$-dependent solution,  we first follow the argument used by \cite{biswas2019estimating} with a given $\vec eta$. We then minimize a class of bounds of $d_{\rm TV}(\pi_t, \pi)$ over the choice of $\vec\eta$ that satisfies (A) and (B). This  leads to a sharper bound than (\ref{eq:biswas}), a special case corresponding to $\vec\eta=0$, which, in general, is not an optimal choice, as shown below.  

We proceed by using the same argument as in \cite{biswas2019estimating} for proving (\ref{eq:biswas}), but using (\ref{cont})
instead of (\ref{eq:hkLf}).  However, when applying (\ref{cont}), we must retain the expression of $H_{k,L}(\cal{X, Y})$, as given by (\ref{eq:hkLf}). (Interested readers are invited to try using (\ref{eq:hkLb}).)
Specifically,
the unbiasedness of (\ref{cont}) implies that, for any $k\ge 1$, 
\beqn\nonumber 
&&  \E[h(X_\pi) - h(X_k)]
= \E\left\{\sum_{j=1}^{J} \left[h(X_{k+jL})-h(Y_{k+(j-1)L})\right]
- \sum_{j\ge 1}\eta_j \Delta_{k,j}\right\}\\
&=&\hskip -3mm \E\left\{\sum_{j\ge 1} \left[h(X_{k+jL})-h(Y_{k+(j-1)L})\right]1_{\{j\le J\}}- \sum_{j\ge 1} \eta_j \left[h(X_{k+jL})-h(Y_{k+jL})\right]\right\}\label{key}\\
&=& \hskip -3mm
\E\left\{\sum_{j\ge 1} h(X_{k+jL})[1_{\{j\le J\}}-\eta_j]+ \sum_{j\ge 1} h(Y_{k+jL})[\eta_j- 1_{\{j+1\le J\}}] - h(Y_k)1_{\{0< J\}}\right\}.\nonumber
\eeqn
The interchanges of sum and expectation  in the (infinite) sums hold under assumption (B) and the additional assumption that  the $h$ function is  bounded. To compute the total variation distance, let $h\in \calh=\{h:  \sup_{x}|h(x)|\le 1/2\}$,  as in \cite{biswas2019estimating}.  Consequently, (\ref{key}) implies 
\beqn\nonumber
d_{\rm TV}(\pi_k, \pi)
&\le &  \frac{1}{2}\left\{\sum_{j\ge 1} \E|1_{\{j\le J\}}-\eta_j|+ \sum_{j\ge 1}\E|\eta_j- 1_{\{j\le J-1\}}|+\Pr(0< J) \right\}\nonumber \\
&=&\sum_{j\ge 1} \E|1_{\{j\le \tilde J\}}-\eta_j|+0.5\Pr(J>0),
\label{newbound}
\eeqn
where $\tilde J=J-\xi$ and  $\xi \sim Bernoulli(0.5)$  is  independent of $J$.  Note that the support for $\tilde J$ is $\{-1, 0, 1, \ldots \}$. Set
\beqn
S_j&=&\Pr(\tilde J\ge j)=\Pr(J>j)+0.5\Pr(J=j), \quad {\rm for\ any}\ j \ge 0. 
\eeqn
Recall that for any given random variable $V$, $\min_{U\perp V} E|V-U|=E|V-m_V|$, where $m_V$
is a median of $V$, and the notation $\min_{U\perp V}$ means to minimize over all $U$ that are independent of $V$.  Hence, in order
to minimize (\ref{newbound}) over $\vec{\eta}$, we should set $\eta_j$ to be the median of the Bernoulli random variable $1_{\{j\le \tilde J\}}$, that is, $\eta_j=1_{\{S_j>0.5\}}$. 

Let $m_{\tilde J}$ be \textit{the smallest integer} median of $\tilde J$. Then, for any $j>m_{\tilde J}$, 
$S_j=1-\Pr(\tilde J < j) \le 1- \Pr(\tilde J \le  m_{\tilde J})\le 1/2$ because $\Pr(\tilde J \le m_{\tilde J})\ge 1/2$, by the definition of $m_{\tilde J}$, implying $\eta_j=0$. Therefore, we know the maximal number of nonzero $\eta_j$ cannot exceed $m_{\tilde J}$. However, other than the ideal case with $\Pr(J=0)=1$, $m_{\tilde J}$  can be zero, but not $-1$, because $\Pr(\tilde J=-1) = 0.5\Pr(J=0)<0.5$. This automatically implies that  condition (B)  is trivially satisfied. For this  choice of $\vec{\eta}$,  (\ref{newbound}) yields our new bound for $d_{\rm TV}(\pi_k, \pi)$:
\beqn
B_{k,L} &=& \sum_{j\ge 1} \min\{S_j, 1-S_j\} +0.5\Pr(J>0)\\
&=& \sum_{j\ge 1} \min\left\{\Pr(J\ge j), \Pr(J\le j)\right\}.
\label{eq:neat}
\eeqn
In deriving the last equality, we use the fact that $S_j+0.5 \Pr(J=j)=\Pr(J\ge j)$ and 
$1-S_j+0.5\Pr(J=j)=\Pr(J\le j)$, and $\Pr(J>0)=\sum_{j\ge 1}\Pr(J=j)$.

\subsection{Understand and Compare the Bounds}
It is immediate from expression (\ref{eq:neat}) that our new bound cannot exceed the bound of \cite{biswas2019estimating}, as given in (\ref{eq:biswas}), because (\ref{eq:neat}) obviously cannot exceed 
$\sum_{j\ge 1} \Pr(J\ge j)$, which is $\E[J]$. 
The next result reveals alternative forms for the new bound, providing  additional insights, including the optimality of the choice  $\eta_j=1_{\{j\le m_{\tilde J}\}}$. 
\begin{theorem}\label{th:newbound}
Under the same regularity conditions as in \cite{biswas2019estimating}, we have 
\beqn
\hskip -1in 
B_{k,L}&=&
 \operatorname{E}|\tilde J_{k,L}-m_{\tilde J_{k,L}}|+\Pr(J_{k,L}>0)  - 0.5\label{eq:medj} \\
&=&\operatorname{E}|J_{k,L}-m_{J_{k,L}}|+ \Pr(J_{k,L}>0)- S_{k, L} \label{eq:medtj}\\
&=& 0.5 \sum_{j\ge 1}[1-|\Pr(\tau>k+(j+1)L)+\Pr(\tau>k+jL)-1| ] \nonumber \\
&+& 0.5\Pr(\tau>k+L), \label{eq:newb} 
\eeqn
where $S_{k, L}=\max\{\Pr(J_{k,L}>m_{J_{k,L}}),\Pr(J_{k,L}<m_{J_{k,L}})\}\le 0.5$, and
$m_{\tilde J_{k,L}}$ and $m_{J_{k,L}}$ are  the smallest integer medians of $\tilde J_{k,L}$ and $J_{k,L}$, respectively.
\end{theorem}

\begin{proof}
To reduce the notation overload, we drop the $k,L$  for $J$, $\tilde J$, $m_J$, and $m_{\tilde J}$. We have already established that the optimal $\vec{\eta}$ must be of the form $\eta_j=1_{\{j\le m \}}$, for some $m\ge 0$. Note here that the use of $m=0$ permits $\vec{\eta}=0$ because $j\ge 1$. This is also consistent with setting $\eta_0=1$. We  can minimize the right-hand side of (\ref{newbound}) with respect to such a class, that is, with respect to the choice of $m$.  However, it is easy to see 
that 
\beqn
\sum_{j\ge 1} \E|1_{\{j\le \tilde J\}}-\eta_j|&=&\sum_{j\ge 0} \E|1_{\{j\le \tilde\nonumber J\}}-1_{\{j\le m\}}|-\E[1-1_{\{0\le \tilde J\}}]\\
&=&\sum_{j\ge 0} \E\left[1_{\{\min\{\tilde J, m\}<j\le 
\max\{\tilde J, m\}}\right]-\Pr(\tilde J=-1)\nonumber \\
&=&\E\left[\max\{\tilde J, m\}-\min\{\tilde J, m\}\right]-0.5\Pr(J=0)\nonumber 
\\\label{eq:minm}
&=&\E|\tilde J- m|-0.5\Pr(J=0).
\eeqn
It is clear from (\ref{eq:minm}) that the optimal $m$ must be an integer median of $\tilde J$, and  we choose the smallest one,  $m_{\tilde J}$. 
With this choice of $\vec{\eta}$, substituting (\ref{eq:minm}) into (\ref{newbound}) yields
the expression 
\beqn
B_{k,L}
&= & \E|\tilde J- m_{\tilde J}|-0.5\Pr(J=0)+0.5\Pr(J>0)\\ 
&=&\Pr(J>0)+\E|\tilde J-m_{\tilde J}|- 0.5,
\label{newbound11}
\eeqn
which proves (\ref{eq:medj}). 

In order to prove (\ref{eq:medtj}), we start from (\ref{eq:neat}). Let 
\beqn 
G(j)= \Pr(J\le j)- \Pr(J\ge j)=
\Pr(J< j)- \Pr(J> j),
\eeqn 
for $j \ge 0$.
Then, it is easy to verify that $G(j)$ is a monotone increasing function, which means $G(j)-G(m_J)$ share the same sign with $j-m_J$, for all $j\ne m_J$.   It follows  that the sum in (\ref{eq:neat}) can be decomposed into three parts, $A=\sum_{j=1}^{m_J-1}  \Pr(J\le j)$, $B=1_{\{m_J>0\}}\min\{\Pr(J\le m_J), \Pr(J\ge m_J)\}$, and $C=\sum_{j\ge m_J+1}  \Pr(J\ge j)$. When $m_J=0$, $C=\E[J]$, $B=0$ because  $1_{\{m_J>0\}}=0$, and  $A=0$  by convention because  $m_J-1<1$. If $p_j=\Pr(J=j)$, then  it is easy to see that whenever $m_J\ge 1$,
\beqn
A&=&\sum_{j=1}^{m_J-1} \sum_{h=1}^{j}p_h + (m_J-1)p_0 =\sum_{h=1}^{m_J-1}\sum_{j=h}^{m_J-1}p_h+(m_J-1)p_0\nonumber\\
&=& \sum_{h=1}^{m_J-1}(m_J-h)p_h +(m_J-1)p_0 =\sum_{h=0}^{m_J-1}(m_J-h)p_h - p_0;\label{eq:parta}\\
C&=&\sum_{j=m_J+1}^\infty  \sum_{h=j}^\infty p_h=\sum_{h=m_J+1}^\infty  \sum_{j=m_J+1}^h p_h=\sum_{h=m_J+1}^\infty  (h-m_J) p_h.
\label{eq:partc}
\eeqn
Noting that $(m_J-h)p_h=0$ when $h=m_J$, we see that when $m_J\ge 1$,
\beqn\nonumber
\hskip -1in A+B+C&=&\E|J-m_J|-p_0+\min\{\Pr(J\ge m_J),\Pr(J\le m_J)\}\\
&=&\E|J-m_J|+\Pr(J>0)+\min\{\Pr(J\ge m_J),\Pr(J\le m_J)\}-1\nonumber\\
&=&\E|J-m_J|+\Pr(J>0)-\max\{\Pr(J< m_J),\Pr(J >m_J)\}. \label{eq:abc}
\eeqn
When $m_J=0$, $A=B=0$, and $C=\sum_{h\ge 1}hp_h=\E[J]$, which is (\ref{eq:medtj}) because $S_{k, L}=\Pr(J>0)$, cancelling exactly the $\Pr(J>0)$ term. This completes the proof of (\ref{eq:medtj}).

The proof of (\ref{eq:newb}) also follows from (\ref{eq:neat}), using the identities $\max\{a, b\}=0.5[a+b+|a-b|]$ and   $\Pr(J\ge j)+\Pr(J\le j) = 1+\Pr(J=j)$, for any $j$. This leads to 
\beqn
&&\sum_{j\ge 1} \min\left\{\Pr(J\ge j), \Pr(J\le j)\right\}  \nonumber \\
&=&0.5\sum_{j\ge 1}\left[ 1+\Pr(J=j)- |\Pr(J\ge j)-1 + \Pr(J>j)|\right]
\nonumber\\
&=& 0.5\sum_{j\ge 1}\left[ 1- |\Pr(J> j)+\Pr(J>j-1)-1|\right] .
+0.5\Pr(J>0).\nonumber 
\eeqn
Expression (\ref{eq:newb}) then follows because 
$\{J>j\}=\{\tau>k+(j+1)L\}.$
\end{proof}

The above result tells us that, whenever $m_J=0$, our bound is identical to the one given by \cite{biswas2019estimating}.  From (\ref{eq:neat}), the two bounds are the same if and only if $G(1) \ge 0$,
which is the same as $2 p_0\ge 1- p_1$, where $p_k=\Pr(J=k)$. Clearly, this inequality is satisfied when $m_J=0$, that is, when $p_0\ge 1/2$. It also implies that $m_J\le 1$, because for any $j<m_J$,
\beqn
G(j)=\Pr(J<j)+\Pr(J\le j)-1\le 2\Pr(J<m_J)-1<0,\label{eq:gproof}
\eeqn
as $\Pr(J\le m_J-1)<0.5$, because $m_J$ is the smallest integer median. Therefore, we have the following theorem.
\begin{theorem}\label{th:iff} Under the same regularity conditions as those of Theorem~\ref{th:newbound}, a sufficient and necessary condition for the bound in Theorem~\ref{th:newbound} to equal $\operatorname{E}[J]$ is  $2 p_0\ge 1- p_1$. 
\end{theorem}



\noindent {\it Remark 3} Theorem \ref{th:iff} implies that $m_J=0$ is a sufficient condition and $m_J\le 1$ is a necessary condition for the two bounds to be the same.  However, the condition $m_J=1$ itself is  not sufficient.  

\noindent {\it Remark 4} An intriguing new insight provided by bound (\ref{eq:medtj}) is that not only the average coupling time matters, but the variation of the coupling time is important too.  The $S_{k, L}$ term also suggests that even the symmetry
matters, because $S_{k, L}$ achieves its maximum when the distribution is symmetrical locally around the median.

Let $\zeta=\tau-t$, which is the number of steps needed  after time $t$ in order to couple  (assuming  the coupling has not already happened by time $t$). Then, the sufficient and necessary condition in Theorem~\ref{th:iff} is the same as 
\beq \Pr(\zeta\le L)\ge \Pr(\zeta>2L), \label{eq:skew}\eeq
suggesting that the new bound is more useful when the distribution of $\zeta$ places  more mass on the right side of the coupling interval $(L, 2L]$ than it does on its left side, that is, when (\ref{eq:skew}) is violated. The implication is that the improvement of the new bound, if any, will more likely come from those situations where either $t$ is small or $\tau$ is large (for fixed $L$), that is, when the mixing is poor.   

\section{Estimation and Practical Implementation}\label{sec:estim}
\setcounter{equation}{0}

We assume that $Q >1$ coupled processes $\{(X_t^{(q)},Y_t^{(q)}): \; 1\le q \le Q\}$ are run in parallel and that, for all $1\le q \le Q$,  ${\cal{X}}^{(q)}=\{X_k^{(q)}\}_{k\ge 1}$ and ${\cal{Y}}^{(q)}=\{Y_k^{(q)}\}_{k \ge 1}$ have been successfully  L-coupled.  The latter implies that the chains ${\cal{X}}^{(q)}$ are run $L$ more steps than the  ${\cal Y}^{(q)}$ chains,  and there exists a stopping time  $\{\tau^{(q)}: q=1, \ldots, Q\}$ such that $X_{t+L}^{(q)} = Y_{t}^{(q)}$, for all $t \ge \tau^{(q)}$. 

\subsection{Control Variate Estimators}

We  work with  a   modified version of  \eqref{eq:hkLb} that incorporates control variates:
\beqn
H^{*(q)}_{k,L} ({\cal X}^{(q)},{\cal Y}^{(q)}) &=& h\left(X_{k+J_{k,L}^{(q)}L}^{(q)}\right) + \sum_{j=0}^{J_{k,L}^{(q)}-1} \left[h(X_{k+jL}^{(q)})-h(Y_{k+jL}^{(q)})  \right]  \nonumber \\
&-& \sum_{j = 0}^{m_{\tilde J_{k,L}^{(q)}} }\left [ h(X_{k+jL}^{(q)})-h(Y_{k+jL}^{(q)}) \right ], 
\label{hksc}
\eeqn
where $m_{\tilde J}$ denotes the smallest integer median of $\tilde{J}$. Henceforth, in order to simplify the notation, we  use   
$m_{k,L}^{(q)}$ and $\tilde m_{k,L}^{(q)}$ to denote  $m_{J_{k,L}^{(q)}}$, and $m_{\tilde J_{k,L}^{(q)}}$, respectively.  An unbiased estimator for  $H^{*(q)}_{k,L} ({\cal X}^{(q)},{\cal Y}^{(q)})$ is straightforward to produce, but additional care must be paid to maintain the independence between the estimator for $\tilde m_{k,L}^{(q)}$ (or $m_{k,L}^{(q)}$)  and $({\cal X}^{(q)},{\cal Y}^{(q)})$. To satisfy the latter, we construct the unbiased estimators  $m_{k,L}^{(q)}$ and $\tilde m_{k,L}^{(q)}$ from all coupled processes but the $q$th one, as described in Algorithm \ref{algo1}.

\begin{algorithm}
\caption{Algorithm for computing $m_{k,L}^{(q)}$ and $\tilde m_{k,L}^{(q)}$for a fixed $k$ and all $q \in \{1,2,\ldots, Q\}$.}\label{algo1}
1.  Compute $J_{k,L}^{(q)}$\;
2. Sample independently $\zeta^{(q)} \sim Bernoulli(0.5)$  and set 
$\tilde J_{k,L}^{(q)} = J_{k,L}^{(q)} - \zeta^{(q)}$\;
3. Set $m_{k,L}^{(q)}=\lfloor med(\{J_{k,L}^{(h)}: \; 1\le h \le Q, \; h\ne q\}\rfloor$ and
$\tilde m_{k,L}^{(q)}=\lfloor med(\{\tilde J_{k,L}^{(h)}: \; 1\le h \le Q, \; h\ne q\}\rfloor,$
where $med(A)$ denotes the median of the values in set $A$ and $\lfloor \cdot \rfloor$ is the floor function.
\end{algorithm}

When $L=1$,  in order to reduce the variance of the unbiased estimator $H_k$ in (\ref{hk}), \cite{jacob2020unbiased} recommend using the time-averaging estimator 
$$H_{k:r}({\cal X},{\cal Y})={1\over r-k+1}\sum_{t=k}^r H_t({\cal X},{\cal Y}).$$    We follow the same strategy, and consider  the time-averaging version of (\ref{eq:hkLb}) 
\beqn
H_{k:r;L}^{(q)}({\cal X}^{(q)},{\cal Y}^{(q)}) &=& 
 \frac{1}{r-k+1}
\sum_{t=k}^r h(X^{(q)}_{t+J_{t,L}^{(q)} L}) \nonumber \\
&+& \frac{1}{r-k+1}
\sum_{t=k}^r \sum_{j=0}^{J_{t,L}^{(q)}-1} \left[h(X_{t+jL}^{(q)})-h(Y_{t+jL}^{(q)})\right], 
 \label{eq:ta-hkL}
\eeqn
and the average estimator that includes the control-variate swindle is then
\beq
H_{k:r;L}^{*(q)}({\cal X}^{(q)},{\cal Y}^{(q)})  = H_{k:r;L}^{(q)}({\cal X}^{(q)},{\cal Y}^{(q)})  -
\frac{1}{r-k+1} \sum_{t=k}^r \left \{  \sum_{j = 0}^{\tilde m_{k,L}^{(q)} }\left [ h(X_{t+jL}^{(q)})-h(Y_{t+jL}^{(q)}) \right ] \right\}.
\label{eq:ta-hkLc}
\eeq
Note that the original versions are obtained when $k=r$.  Because each term in the control-variate term above, 
$h(X_{t+jL}^{(q)})-h(Y_{t+jL}^{(q)})$, has mean zero, we expect that the gain from the control variate  swindle diminishes when $r$ increases owing to the law of large numbers, leading to the overall control-variate term approaching zero.  We see this phenomenon in Section~\ref{sec:prac}.

\subsection{Estimating the Total Variation Bound}
When estimating $B_{k,L}$, we can use  \eqref{eq:medj}, \eqref{eq:medtj}, or \eqref{eq:newb}. In  our numerical experiments we use
\eqref{eq:medtj}, along the steps described in Algorithm \ref{algo2}. 
\begin{algorithm}
\caption{Algorithm for estimation of $B_{k,L}$, for any given $k$ and $L$. }\label{algo2}
1. Compute $J_{k,L}^{(q)}$ and $m_{k,L}^{(q)}$, for all $q=1,\ldots,Q$ \; 
2. Compute the empirical means
$$e_{k,L} = {1\over Q} \sum_{q=1}^Q \left | J_{k,L}^{(q)} - m_{k,L}^{(q)} \right |, \quad p_{k,L} = {1\over Q}  \sum_{q=1}^Q \one_{\{J_{k,L}^{(q)} > 0\}} $$
$$a_{k,L}= {1\over Q}  \sum_{q=1}^Q \one_{\{J_{k,L}^{(q)} > m_{k,L}^{(q)}\}}, \quad 
b_{k,L}= {1\over Q}  \sum_{q=1}^Q \one_{\{J_{k,L}^{(q)} < m_{k,L}^{(q)}\}}; $$ 
3. Compute $$\hat B_{k,L} = e_{k,L} + p_{k,L} -  a_{k,L} \vee b_{k,L},$$ where $a \vee b$ denotes the maximum between $a$ and $b$.
\end{algorithm}

In the next section, we investigate the performance of the control variate swindle and compare the new total variation bound with  \eqref{eq:biswas} provided by \cite{biswas2019estimating}.

\section{Examples and Illustrations}\label{sec:prac}
\setcounter{equation}{0}

\subsection{A Theoretical Comparison of the Bounds: The Geometric Case}

The distribution of the coupling time $\tau_L$ is, in general, unknown.  However, there is one instance in which the distribution of $\tau_L$  is  exactly geometric. Specifically, when coupling two independent Metropolis samplers,  the maximal coupling procedure  uses the same proposal for both transition kernels, and coupling occurs when both chains accept it.  The tractability of derivations in the geometric case allow us to better understand theoretically the relationship between the bound in \cite{biswas2019estimating} and ours.

We consider the case in which $\delta=\tau-(L-1)\sim Geo(p)$. Because $J=\max\left\{0,  \ceil{\frac{\delta-k-1}{L}}\right\}$, we see that 
\beqn
   \Pr(J=0) &=& \Pr(\delta\le k+1)=1- q^{k+1}, \nonumber\\
    \Pr(J>j) &=&  \Pr(\delta> k+1+Lj) = q^{k+1+Lj}, \quad j=0, 1, \ldots,  \label{eq:geo}
\eeqn
where $q=1-p$. That is, the distribution of $J$ is a mixture of (i) the Dirac point measure  $\delta_{\{0\}}$ with mixture proportion $1-q^{k+1}$, and (ii) a geometric distribution with probability of success $1-q^L$ with weight $q^{k+1}$. This implies immediately that the bound given in \cite{biswas2019estimating} has the expression (\ref{eq:j-mean}). 

For our new bound, in this case, it is easier to use expression (\ref{eq:neat}) directly.  Let $m$ be the largest integer such that $\Pr(J\ge m) \ge \Pr(J\le m)$; that is,  $m$ is the largest integer that ensures
\beqn\label{eq:smallm}
q^{k+1+L(m-1)}+q^{k+1+Lm}\ge  1 \quad \Longleftrightarrow \quad    m=\left\lfloor \frac{L-k-1}{L}-\frac{\log[1+q^{L}]}{L\log(q)}  \right\rfloor.
\eeqn
Clearly, when $m\le 0$, our $B_{k, L}(p)$ is the same as the old bound (\ref{eq:j-mean}).  When $m\ge 1$,
we have
\beqn\label{eg:newb}
 B_{k, L}(p) &=& \sum_{j=1}^{m} 
\Pr(J\le j)+\sum_{j= m+1} 
\Pr(J\ge j) \nonumber\\ 
\{\rm{by}\  (\ref{eq:geo})\} &=& \sum_{j=1}^{m} 
\left[1-q^{k+1+Lj}\right]+\sum_{j =m+1}^{\infty} q^{k+1+L(j-1)}\nonumber\\
&=& m - \frac{q^{k+1+L}[1-q^{mL}]}{1-q^L}+ \frac{q^{k+1+mL}}{1-q^L}\nonumber
\\
&=& m - \frac{q^{k+1+L}[1-q^{mL}-
q^{(m-1)L}]}{1-q^L}.
\eeqn
We can also compute $B_{k, L}(p)$ directly from (\ref{eq:newb}) as an infinite sum
\beqn
B_{k,L}(p)= 0.5 \sum_{j\ge 1}[1-|q^{k+1+L(j-1)}+q^{k+1+Lj}-1| ] + 0.5q^{k+1}.
\label{newbound-r}
\eeqn

Figure \ref{fig:geo} compares the bounds (\ref{eq:j-mean})  (dashed line) and \eqref{newbound-r}  (solid line) for different values of $p$, $L$, and $t$. One can see that the new bound is sharper, and that only  for larger values of $p$, corresponding to fast-mixing chains, are the two bounds  indistinguishable. The horizontal line  in Figure \ref{fig:geo}  marks the obvious bound, because $d_{\rm TV} \le 1$. Note too that for very small values of $p$, both bounds are vacuous, but the new bound has a larger range for being  nonvacuous.

\begin{figure}
\centering
\includegraphics[width=0.95\textwidth]{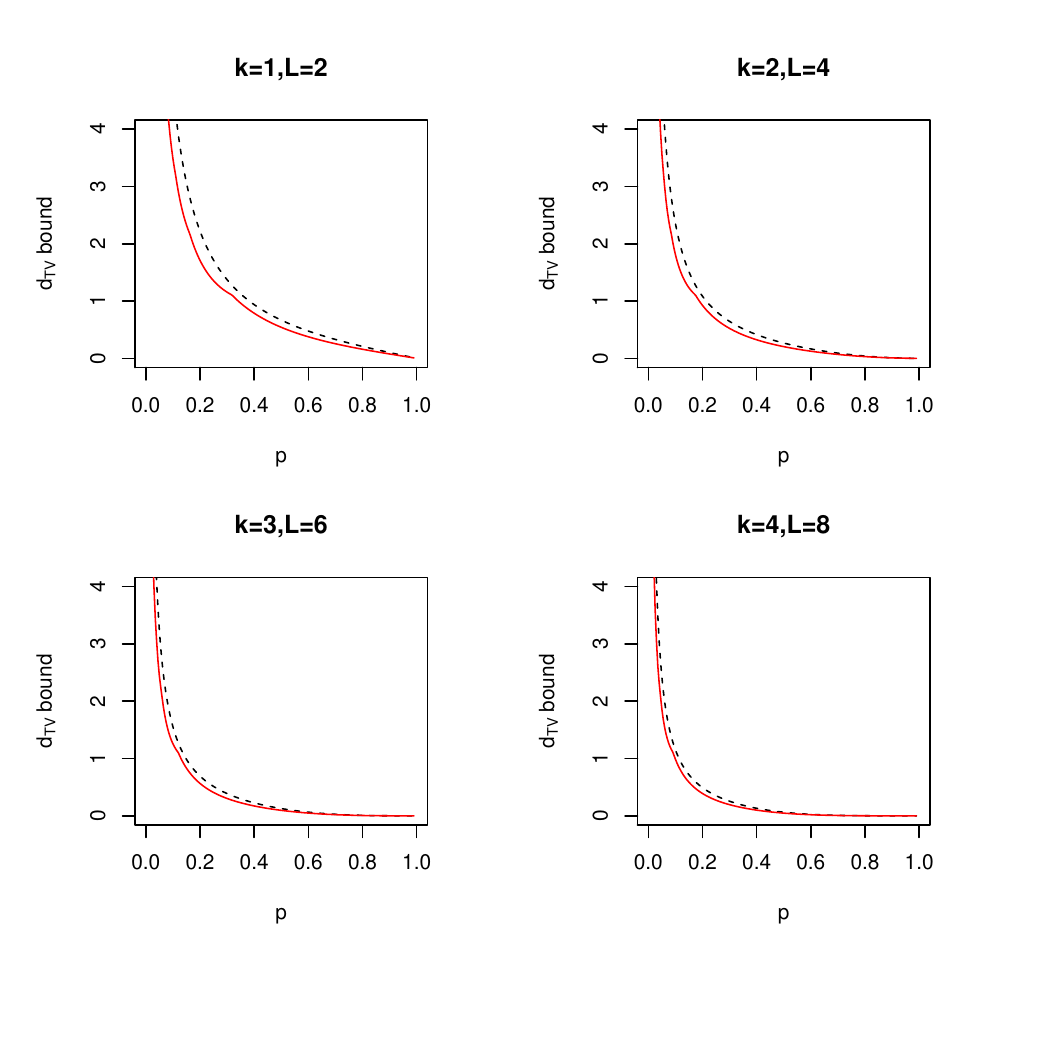}
\caption{Comparison of the bound \eqref{eq:biswas} provided by \cite{biswas2019estimating} (dashed  line), and the new bound  given in \eqref{newbound11} (solid  line).  Note that for small values of $p$, both bounds are vacuous.}
\label{fig:geo}
\end{figure}

\begin{figure}
\centering
\includegraphics[width=0.95\textwidth]{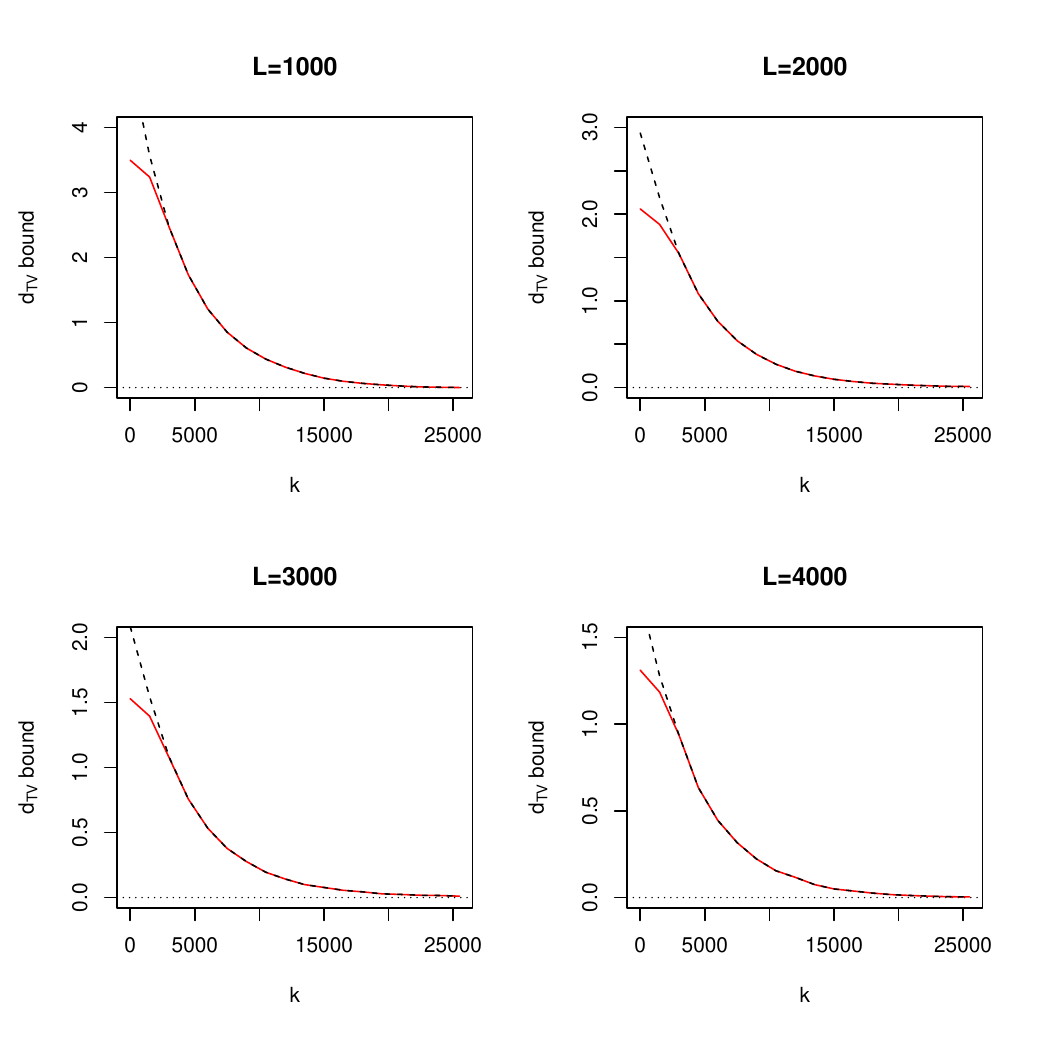}
\caption{Ising Model: Comparison of TV bounds for the PT algorithm for SGS for $L \in \{1000,2000,3000,4000\}$. The  dashed  line shows the bound  \eqref{eq:biswas} derived in \cite{biswas2019estimating},  and the solid  line shows the new bound  given in \eqref{newbound11}. 
}
\label{fig:ising1-4k}
\end{figure}

The simulations in the remaining two examples rely on the 
{\tt{unbiasedmcmc}} package of Pierre Jacob,  available from:\\
 {\tt https://github.com/pierrejacob/unbiasedmcmc/tree/master/vignettes}.  Additional programs for implementing the new ideas in this paper are available as supplemental material from the authors.

\subsection{An Empirical Comparison of the Bounds: Ising Model}

The Ising model example follows the setup in  \cite{biswas2019estimating}. We consider a $32 \times 32$ square lattice of  pixels  with values in $\{-1,1\}$, and with periodic boundaries. A state of the system  is then  $x \in \{-1,1\}^{32 \times 32}$, and the target probability is defined as $\pi_\beta(x) \propto \exp(\beta \sum_{i\sim j} x_ix_j)$, where $i \sim j$ means that $x_i$ and $x_j$ are pixel values in neighboring sites. This illustration uses the parallel tempering  algorithm \citep[PT, see][]{swendsen1986replica} coupled with a single site Gibbs (SSG) updating. It is known that larger values of $\beta$ increase the dependence between neighboring sites, and that this ``stickiness" leads to slow mixing of the SSG.  The target of interest corresponds to $\beta_0=0.46$ and  we use 12 chains, each corresponding to a different $\pi_\beta (x)$, with $\beta$ values equally spaced between $0.3$ and $\beta_0=0.46$. 
Figure \ref{fig:ising1-4k} shows the total variation bounds, where that provided by \eqref{eq:biswas} is shown as a dashed line and that from  \eqref{newbound11} is shown as a solid line. The bounds are derived for $1\le k <25,000$  and $L\in\{1000, 2000, 3000, 4000\}$. For smaller values of $L$, the patterns are similar, but TV bounds are larger for smaller values of $k$. The new bound, \eqref{newbound11},  is computed from  
$Q=50$ parallel runs, and is averaged over    $20$ independent replicates, while  \eqref{eq:biswas} is averaged over $1000$ independent replicates of a single coupled process.

Although the numerical results confirm that our new bound never exceeds the bound of  \cite{biswas2019estimating}, unfortunately, in this case, the improvement from our bound is visible only when it is not needed, that is, when both bounds exceed one.  Whereas this is a disappointment for our effort to improve the bound with a real gain, it is  good news for practitioners, because the bound in  \cite{biswas2019estimating} is a bit simpler to use.

\subsection{Comparing Bounds and Estimators: A Logistic Regression Example}

To compare  the bounds and the unbiased estimators,
we follow \cite{biswas2019estimating} and consider a Bayesian logistic regression model for the German credit data of \cite{lichman2013uci}. The data consist of $n=1000$ binary responses, $\{Y_i: \; 1\le i \le n\}$ 
and $d=49$ covariates, $\{x_i \in \RR^d; \: 1\le i \le n\}$. The response $Y_i$  indicates whether the $i$th individual is fit to receive credit ($Y_i=1$) or not ($Y_i =0$). The logistic regression model frames the probabilistic dependence between the response and covariate as $\Pr(Y_i=1|x_i)= [1+\exp(- x_i^T\beta)]^{-1}$. The prior is set to $\beta \sim N(0,10\I_d)$. Sampling from the posterior distribution is done using the P\'olya-Gamma sampler of \cite{polson2013bayesian}, using the R programs made available by  \cite{biswas2019estimating} at {\tt{
https://github.com/niloyb/LlagCouplings.}} In Figure \ref{fig:pg210}, we compare the bound of  \cite{biswas2019estimating} \eqref{eq:biswas} (dashed  line) with our bound \eqref{newbound11} (solid  line).  The bound in \eqref{eq:biswas}  is averaged over 2000 independent replicates. The new bound is computed from running 50 coupled processes in parallel and averaged over 40 replicates, yielding the same number of runs. The difference between the two bounds is apparent for smaller values of $L$ when the new bound is sharper for small values of $k$, but the gain diminishes quickly as $L$ increases.   

\begin{figure}[htp]
\centering
\includegraphics[width=0.95\textwidth]{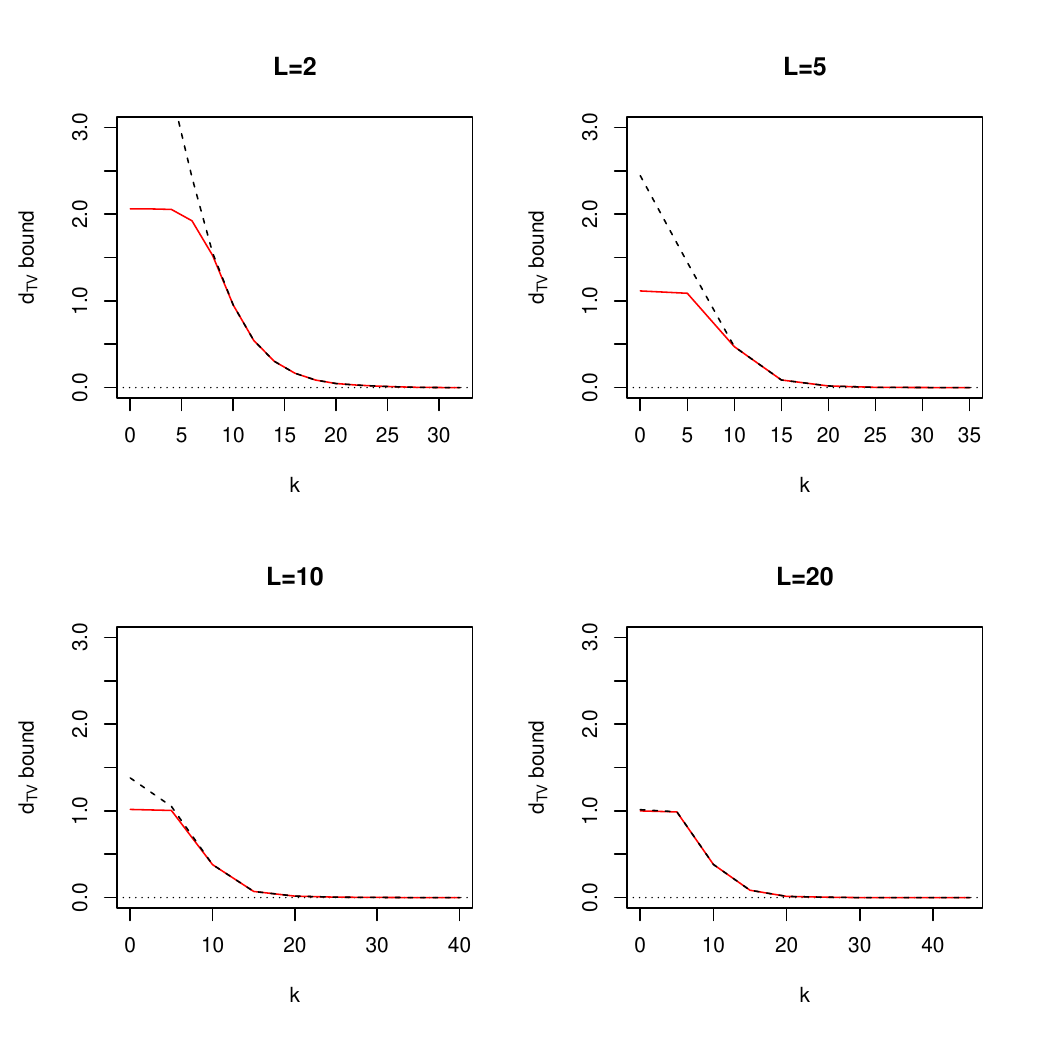}
\caption{German credit Ddta: Comparison of TV bounds for the P\'olya-Gamma sampler for $L \in \{2,5,10,20\}$. The  dashed  line shows the bound  \eqref{eq:biswas} derived in \cite{biswas2019estimating},  and the solid  line shows the new bound  given in \eqref{newbound11}. The bound from \eqref{newbound11}  is obtained from running 50 coupled chains in parallel and averaging over 40 independently replicated  experiments. The bound from \eqref{eq:biswas}  is averaged over 2000 independent replicates.  }
\label{fig:pg210}
\end{figure}

We are also interested in the gains in efficiency for the  Monte Carlo estimators when implementing the control variate swindle. Using 500 independent replicates of a single coupled process with lag $L=5$, we obtain Monte Carlo estimates of the posterior means for the 
regression coefficients.  In Figure  \ref{fig:pgsingle}, we present the relative reduction in variance (RRV), computed as 
RRV$={\rv_{MCCV}(\hat\beta) \over \rv_{MC}(\hat\beta)}$, where $\hat\beta$ is the posterior mean of the regression coefficients, $\beta \in \RR^{49}$, and $\rv_{MC}$ and $\rv_{MCCV}$ denote the estimated Monte Carlo variances of $\hat\beta$ obtained without and with the control variates, respectively. The left panel shows the RRV when using the single-run estimators \eqref{eq:hkLb} and \eqref{hksc}, while the right panel plots the RRV for the mean estimators \eqref{eq:ta-hkL} and \eqref{eq:ta-hkLc}. We see clearly that the gain is  significant for $r=k=5$ (left panels), but diminishes when $k=5$ and $r=30$ (right panels), as  discussed in Section~\ref{sec:estim}.

\begin{figure}[htp]
\centering
\includegraphics[width=0.45\textwidth]{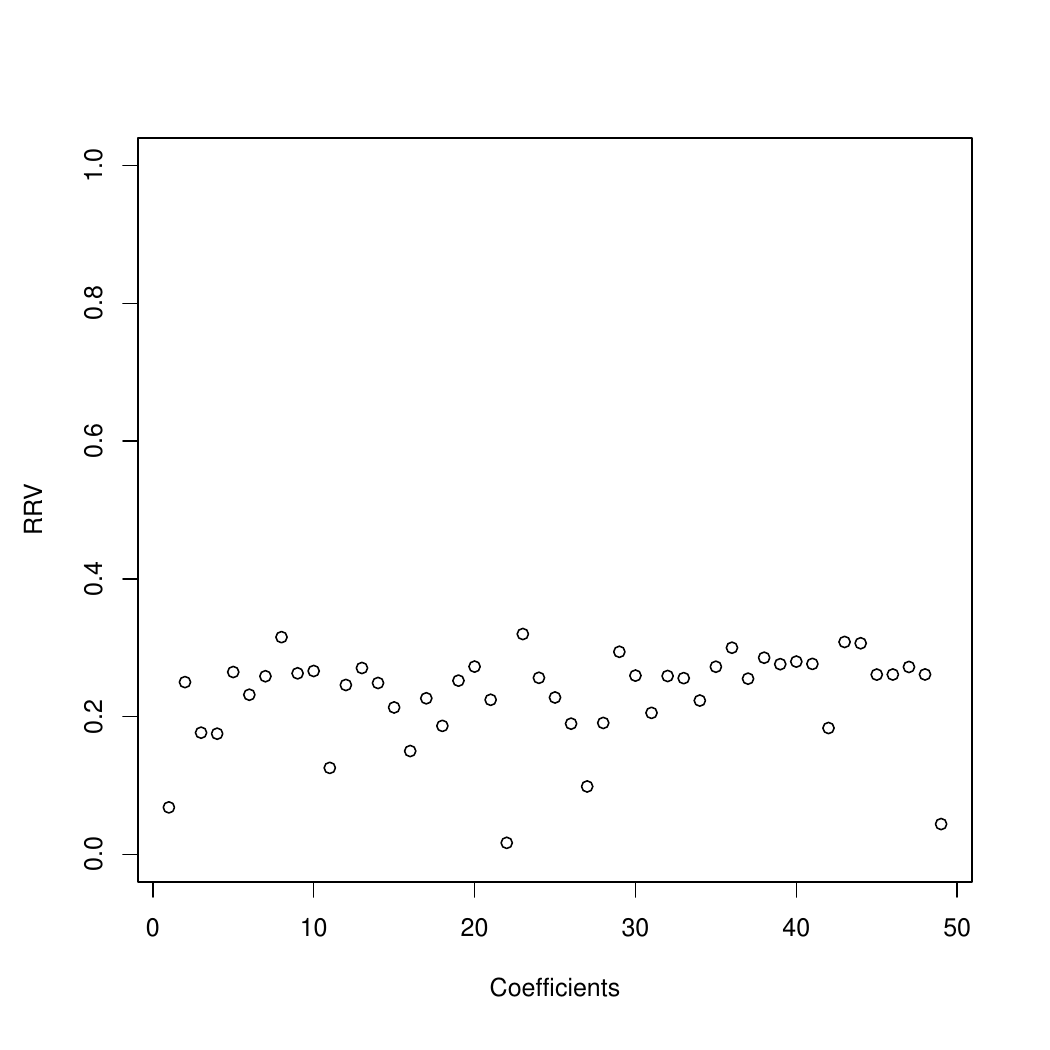}
\includegraphics[width=0.45\textwidth]{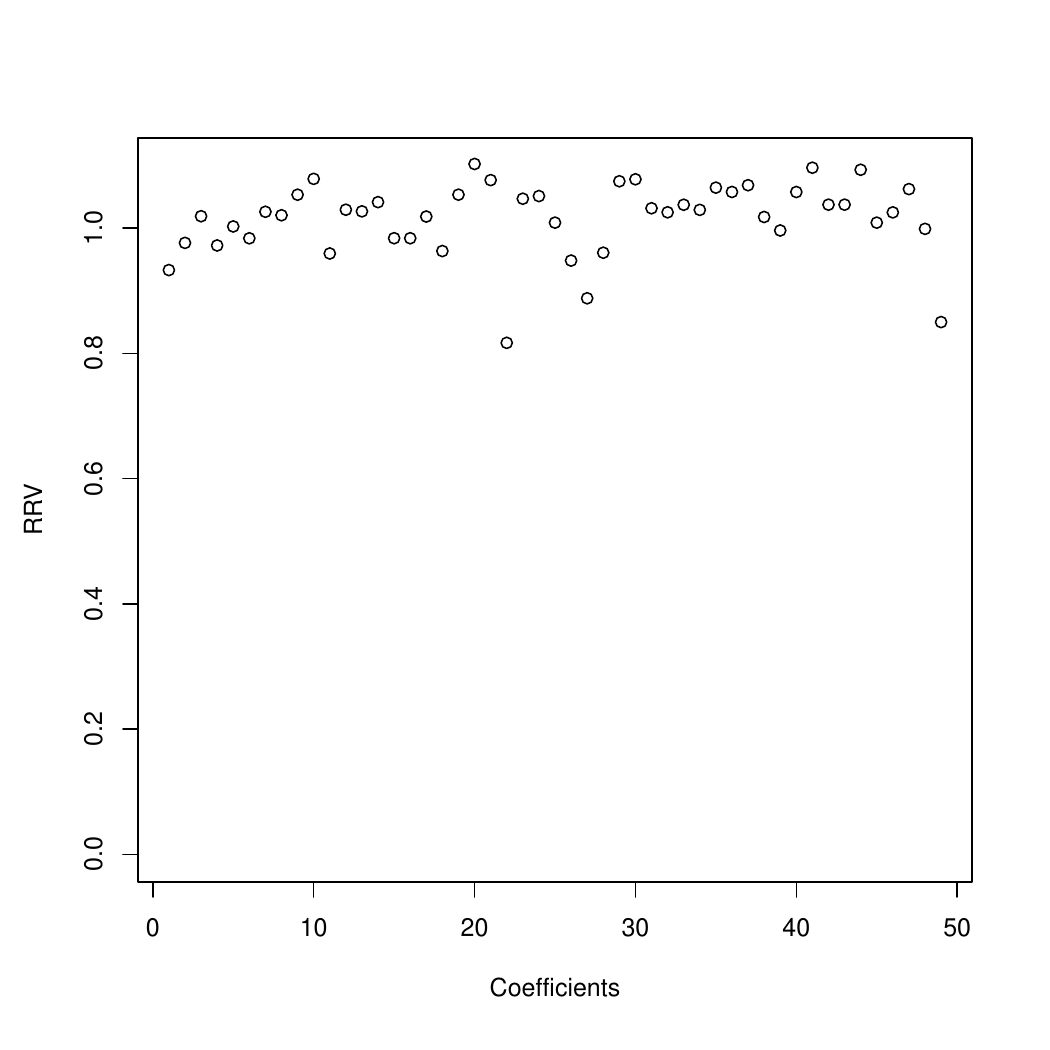}
\includegraphics[width=0.45\textwidth]{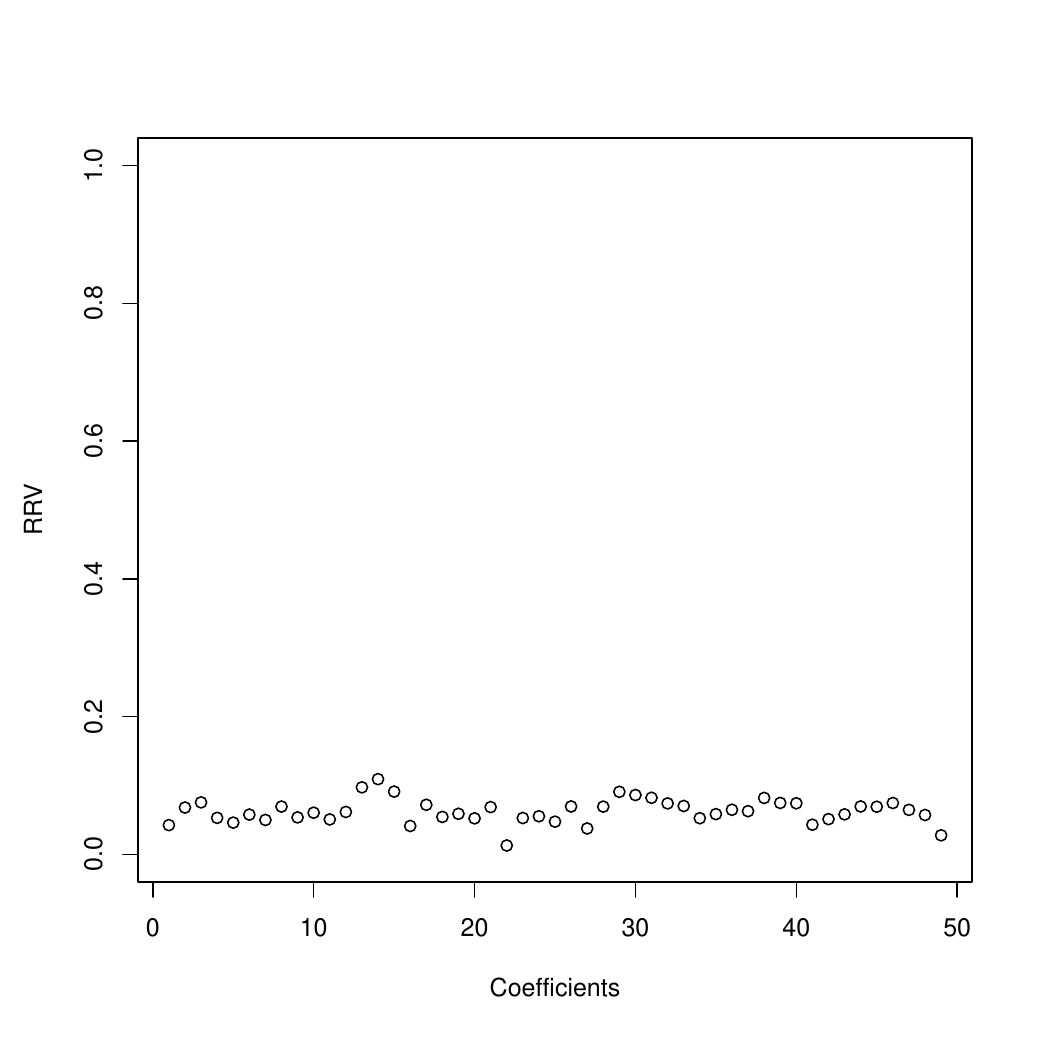}
\includegraphics[width=0.45\textwidth]{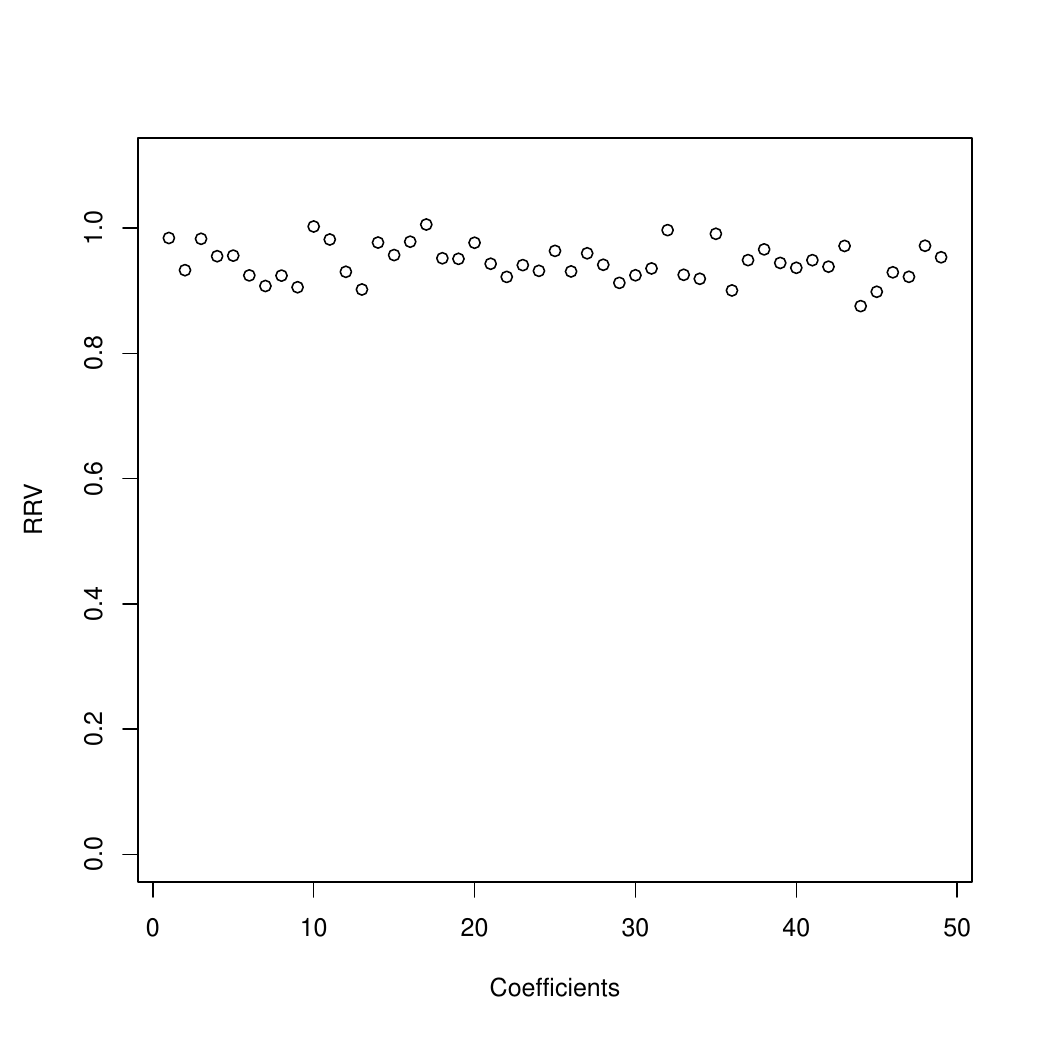}
\caption{German credit data.  Relative  reduction in variance (RRV) for the 49 regression coefficients. 
{\it Top panels}: the lag is $L=5$. {\it Bottom panels}: the lag is $L=20$. 
{\it Left panels}: RRV is obtained from the single estimators without and  with control variates, respectively, using $k=5$ in \eqref{eq:hkLb} and \eqref{hksc}. {\it Right panels}: RRV is obtained from the average estimators without and with control variates, respectively, using $k=5$ and $r=30$ in \eqref{eq:ta-hkL} and \eqref{eq:ta-hkLc}. 
}
\label{fig:pgsingle}
\end{figure}

\section{Can We Do Even Better?} \label{sec:fut}
The idea of L-lag coupling  has opened multiple avenues for future research. The use of control variates is just one of them. Although the practical gain is small or possibly even negative when we take into account the increased computation when computing the control variates, 
the theoretical gain is  intriguing, because  we obtain a theoretically superior bound without imposing any additional assumptions. This naturally raises the question of whether our bound is the best possible without further conditions.  We do not know.  We do not even know how to study  such a question theoretically, because to the best of our knowledge, this is the first
time  a tighter \textit{theoretical bound} has been obtained by a better \textit{empirical estimator}. 
Whereas seeking other more efficient estimators seems to be a natural direction, we must keep in mind that they would likely incur additional computational costs. 

One plausible direction is to go beyond linearly combining mean-zero control variates, although we had no success so far. However, even without seeking better bounds, our current bounds already offer the opportunity to investigate fresh perspectives for optimizing an MCMC kernel using adaptive ideas, and we intend to pursue these in future research. 

\section*{Acknowledgments}  We are grateful to Pierre Jacob for many useful comments and his gracefully patient guidance through the package {\tt unbiasedmcmc}, allowing us to perform the simulations in our study. We also thank Yves Atchad\'e, Tamas Papp, Christopher Sherlock, and Lei Sun for their helpful discussions and comments, and the NSERC of Canada (RVC) and NSF of USA (XLM) for their partial research support.

%

\begin{thebibliography}{}

\bibitem[\protect\citeauthoryear{Berthelsen and M{\o}ller}{Berthelsen and
  M{\o}ller}{2002}]{MR1978833}
Berthelsen, K.~K. and M{\o}ller, J. (2002).
\newblock A primer on perfect simulation for spatial point processes.
\newblock {\em Bull. Braz. Math. Soc. (N.S.)\/}~{\bf 33}, 351--367.
\newblock Fifth Brazilian School in Probability (Ubatuba, 2001).

\bibitem[\protect\citeauthoryear{Biswas, Jacob, and Vanetti}{Biswas
  et~al.}{2019}]{biswas2019estimating}
Biswas, N.,  Jacob, P.~E. and Vanetti, P. (2019).
\newblock Estimating convergence of Markov chains with L-lag couplings.
\newblock In {\em Advances in Neural Information Processing Systems}, pp.\
  7389--7399.

\bibitem[\protect\citeauthoryear{Corcoran and Schneider}{Corcoran and
  Schneider}{2005}]{MR2154791}
Corcoran, J.~N. and Schneider, U. (2005).
\newblock Pseudo-perfect and adaptive variants of the {M}etropolis-{H}astings
  algorithm with an independent candidate density.
\newblock {\em J. Stat. Comput. Simul.\/}~{\bf 75\/}, 459--475.

\bibitem[\protect\citeauthoryear{Corcoran and Tweedie}{Corcoran and
  Tweedie}{2002}]{MR1906013}
Corcoran, J.~N. and  Tweedie, R.~L. (2002).
\newblock Perfect sampling from independent {M}etropolis-{H}astings chains.
\newblock {\em J. Statist. Plann. Inference\/}~{\bf 104\/}, 297--314.

\bibitem[\protect\citeauthoryear{Craiu and Lemieux}{Craiu and
  Lemieux}{2007}]{craiu2007acceleration}
Craiu, R.~V. and Lemieux, C. (2007).
\newblock Acceleration of the multiple-try Metropolis algorithm using
  antithetic and stratified sampling.
\newblock {\em Statistics and Computing\/}~{\bf 17\/}, 109--120.

\bibitem[\protect\citeauthoryear{Craiu and Meng}{Craiu and
  Meng}{2001}]{craiu2001antithetic}
Craiu, R.~V. and  Meng, X.-L. (2001).
\newblock Antithetic coupling for perfect sampling.
\newblock In E.~I. George (Ed.), {\em Bayesian Methods, with Applications to Science, Policy and Official Statistics (Proceedings of the ISBA 2000 conference, Hersonnissos, Crete)}, 99-108.  Luxembourg: Office for Official Publications of the European Communities.

\bibitem[\protect\citeauthoryear{Craiu and Meng}{Craiu and
  Meng}{2005}]{craiu2005multiprocess}
Craiu, R.~V. and  Meng, X.-L. (2005).
\newblock Multiprocess parallel antithetic coupling for backward and forward
  Markov chain Monte Carlo.
\newblock {\em The Annals of Statistics\/}~{\bf 33\/}, 661--697.

\bibitem[\protect\citeauthoryear{Craiu and Meng}{Craiu and
  Meng}{2011}]{craiu2011perfection}
Craiu, R.~V. and  Meng, X.-L.  (2011).
\newblock Perfection within reach: exact MCMC sampling.
\newblock {\em Handbook of Markov Chain Monte Carlo\/}, 199--226.

\bibitem[\protect\citeauthoryear{Craiu and Meng}{Craiu and Meng}{2020}]{cmdisc}
Craiu, R.~V. and  Meng, X.-L.  (2020).
\newblock Discussion of {"Unbiased Markov chain Monte Carlo with couplings"} by
  {Pierre E. Jacob, John O'Leary and Yves F. Atchad\'e}.
\newblock {\em J. Royal Statist. Society, Series B\/}, {\bf  82}, 578--581.

\bibitem[\protect\citeauthoryear{Dobrow and Fill}{Dobrow and
  Fill}{2003}]{MR2016872}
Dobrow, R.~P. and  Fill, J.~A. (2003).
\newblock Speeding up the {FMMR} perfect sampling algorithm: a case study
  revisited.
\newblock {\em Random Structures Algorithms\/}~{\bf 23\/}, 434--452.

\bibitem[\protect\citeauthoryear{Ensor and Glynn}{Ensor and
  Glynn}{2000}]{ensor2000simulating}
Ensor, K.~B. and Glynn, P.~W.  (2000).
\newblock Simulating the maximum of a random walk.
\newblock {\em Journal of Statistical Planning and Inference\/}~{\bf
  85\/}, 127--135.

\bibitem[\protect\citeauthoryear{Glynn}{Glynn}{2016}]{glynn2016exact}
Glynn, P.~W. (2016).
\newblock Exact simulation vs exact estimation.
\newblock In {\em 2016 Winter Simulation Conference (WSC)}, pp.\  193--205.
  IEEE.

\bibitem[\protect\citeauthoryear{Glynn and Heidelberger}{Glynn and
  Heidelberger}{1991}]{glynn1991analysis}
Glynn, P.~W. and Heidelberger, P. (1991).
\newblock Analysis of parallel replicated simulations under a completion time
  constraint.
\newblock {\em ACM Transactions on Modeling and Computer Simulation
  (TOMACS)\/}~{\bf 1}, 3--23.

\bibitem[\protect\citeauthoryear{Glynn and Rhee}{Glynn and
  Rhee}{2014}]{glynn2014exact}
Glynn, P.~W. and  Rhee, C.-h. (2014).
\newblock Exact estimation for Markov chain equilibrium expectations.
\newblock {\em Journal of Applied Probability\/}~{\bf 51}, 377--389.

\bibitem[\protect\citeauthoryear{Heng and Jacob}{Heng and
  Jacob}{2019}]{heng2019unbiased}
Heng, J. and Jacob, P.~E.  (2019).
\newblock Unbiased Hamiltonian Monte Carlo with couplings.
\newblock {\em Biometrika\/}~{\bf 106\/}, 287--302.


\bibitem[\protect\citeauthoryear{Huber}{Huber}{2002}]{sw-hub}
Huber, M.~L. (2002).
\newblock A bounding chain for {S}wendsen-{W}ang.
\newblock {\em Random Structures and Algorithms\/}~{\bf 22}, 43--59.

\bibitem[\protect\citeauthoryear{Huber}{Huber}{2004}]{MR2052900}
Huber, M.~L. (2004).
\newblock Perfect sampling using bounding chains.
\newblock {\em Ann. Appl. Probab.\/}~{\bf 14}, 734--753.

\bibitem[\protect\citeauthoryear{Jacob, Lindsten, and Sch{\"o}n}{Jacob
  et~al.}{2019}]{jacob2019smoothing}
Jacob, P.~E., Lindsten, F. and  Sch{\"o}n, T.~B. (2019).
\newblock Smoothing with couplings of conditional particle filters.
\newblock {\em Journal of the American Statistical Association\/}, {\bf 115}, 721--729.


\bibitem[\protect\citeauthoryear{Atchad{\'e}, Jacob, and O'Leary}{Jacob
  et~al.}{2020}]{jacob2020unbiased}
Jacob,  P.~E., O'Leary, J. and Atchad{\'e},  Y.~F.  (2020).
\newblock Unbiased {M}arkov chain {Monte Carlo} with couplings (with
  discussion).
\newblock {\em J. Royal Statist. Society, Series B\/},  {\bf 82} , 543--600.


\bibitem[\protect\citeauthoryear{Lichman}{Lichman}{2013}]{lichman2013uci}
Lichman, M. (2013).
\newblock Uci machine learning repository, 2013.

\bibitem[\protect\citeauthoryear{Meng}{Meng}{2000}]{meng:multistage-backwards}
Meng, X.~L. (2000).
\newblock Towards a more general {Propp-Wilson} algorithm: {M}ultistage
  backward coupling.
\newblock In N.~Madras (Ed.), {\em {M}onte {C}arlo Methods}, Volume~{\bf 26} of
  {\em Fields Institute Communications}, pp.\  85--93. American Mathematical
  Society.

\bibitem[\protect\citeauthoryear{M{\o}ller}{M{\o}ller}{1999}]{moller:conditional}
M{\o}ller, J. (1999).
\newblock Perfect simulation of conditionally specified models.
\newblock {\em J. Royal Statist. Society, Series B\/}~{\bf 61}, 251--264.

\bibitem[\protect\citeauthoryear{Murdoch and Meng}{Murdoch and
  Meng}{2001}]{murdoch2001}
Murdoch, D.~J. and Meng, X.-L.  (2001).
\newblock Towards perfect sampling for {B}ayesian mixture priors.
\newblock In E.~I. George (Ed.), {\em Bayesian Methods, with Applications to Science, Policy and Official Statistics (Proceedings of the ISBA 2000 conference, Hersonnissos, Crete)}, 381-390.  Luxembourg: Office for Official Publications of the European Communities.

\bibitem[\protect\citeauthoryear{Murdoch and Takahara}{Murdoch and
  Takahara}{2006}]{mur-que}
Murdoch, D.~J. and Takahara, G. (2006).
\newblock Perfect sampling for queues and network models.
\newblock {\em ACM Transactions on Modeling and Computer Simulation\/}~{\bf
  16}, 76--92.

\bibitem[\protect\citeauthoryear{Nelson}{Nelson}{2016}]{nelson2016some}
Nelson, B.~L. (2016).
\newblock `Some tactical problems in digital simulation' for the next 10 years.
\newblock {\em Journal of Simulation\/}~{\bf 10\/}, 2--11.

\bibitem[\protect\citeauthoryear{Polson, Scott, and Windle}{Polson
  et~al.}{2013}]{polson2013bayesian}
Polson, N.~G.,  Scott, J.~G. and Windle, J. (2013).
\newblock Bayesian inference for logistic models using P{\'o}lya--Gamma latent
  variables.
\newblock {\em Journal of the American Statistical Association\/}~{\bf
  108\/}, 1339--1349.

\bibitem[\protect\citeauthoryear{Propp and Wilson}{Propp and
  Wilson}{1996}]{propp-wilson:exact-sampling}
Propp, J.~G. and  Wilson, D.~B. (1996).
\newblock Exact sampling with coupled {M}arkov chains and applications to
  statistical mechanics.
\newblock {\em Random Structures and Algorithms\/}~{\bf 9},
  223--252.

\bibitem[\protect\citeauthoryear{Propp and Wilson}{Propp and
  Wilson}{1998}]{propp-wilson}
Propp, J.~G. and Wilson, D.~B.  (1998).
\newblock How to get a perfectly random sample from a generic {M}arkov chain
  and generate a random spanning tree of a directed graph.
\newblock {\em J. Algorithms\/}~{\bf 27\/}, 170--217.
\newblock 7th Annual ACM-SIAM Symposium on Discrete Algorithms (Atlanta, GA,
  1996).

\bibitem[\protect\citeauthoryear{Stein and Meng}{Stein and
  Meng}{2013}]{stein2013practical}
Stein, N.~M. and Meng, X.-L.  (2013).
\newblock Practical perfect sampling using composite bounding chains: the
  Dirichlet-multinomial model.
\newblock {\em Biometrika\/}~{\bf 100\/}, 817--830.

\bibitem[\protect\citeauthoryear{Swendsen and Wang}{Swendsen and
  Wang}{1986}]{swendsen1986replica}
Swendsen, R.~H. and  Wang, J.-S. (1986).
\newblock Replica Monte Carlo simulation of spin-glasses.
\newblock {\em Physical review letters\/}~{\bf 57\/}, 2607.

\bibitem[\protect\citeauthoryear{Th{\"o}nnes}{Th{\"o}nnes}{1999}]{MR1699662}
Th{\"o}nnes, E. (1999).
\newblock Perfect simulation of some point processes for the impatient user.
\newblock {\em Adv. in Appl. Probab.\/}~{\bf 31\/}, 69--87.

\bibitem[\protect\citeauthoryear{Van~Dyk and Meng}{Van~Dyk and
  Meng}{2001}]{van2001art}
Van~Dyk, D.~A. and  Meng, X.-L. (2001).
\newblock The art of data augmentation (with discussion).
\newblock {\em Journal of Computational and Graphical Statistics\/}~{\bf
  10\/}, 1--50.

\bibitem[\protect\citeauthoryear{Wilson}{Wilson}{1998}]{exact-bibliography}
Wilson, D.~B. (1998).
\newblock Annotated bibliography of perfectly random sampling with \char77arkov
  chains.
\newblock In D.~Aldous and J.~Propp (Eds.), {\em Microsurveys in Discrete
  Probability}, Volume~{\bf 41} of {\em DIMACS Series in Discrete Mathematics
  and Theoretical Computer Science}, pp.\  209--220. American Mathematical
  Society.
\newblock Updated versions appear at {\tt http://www.dbwilson.com/exact/}.

\bibitem[\protect\citeauthoryear{Yu and Meng}{Yu and Meng}{2011}]{yu2011center}
Yu, Y. and  Meng, X.-L. (2011).
\newblock To center or not to center: That is not the question---an
  {Ancillarity--Sufficiency Interweaving Strategy (ASIS)} for boosting {MCMC}
  efficiency (with discussion).
\newblock {\em Journal of Computational and Graphical Statistics\/}~{\bf
  20\/}, 531--570.

\end{thebibliography}
%
\end{document}